\documentclass[12pt]{article}
\usepackage{amsmath}
\usepackage{graphicx,psfrag,epsf}
\usepackage{enumerate}
\usepackage{natbib}
\usepackage{url} 
\usepackage{comment}
\usepackage{tabularx,ragged2e}
\newcommand{\blind}{0}
\usepackage{algorithm}

\usepackage{caption}
\usepackage{subcaption}
\usepackage{diagbox}
\usepackage{pict2e}
\usepackage[english]{babel}
\usepackage[utf8]{inputenc}
\usepackage{relsize}
\usepackage{amsmath,amssymb}
\usepackage[english]{babel}
\usepackage{amsthm}
\newtheorem{prop}{Proposition}[section]
\usepackage[toc,page]{appendix}
\allowdisplaybreaks
\usepackage{parskip}
\usepackage{array,multirow,graphicx}
\newtheorem{theorem}{Theorem}[section]
\newtheorem{corollary}{Corollary}[theorem] 

\usepackage{algorithm,algorithmic}
\graphicspath{ {images/} }
\usepackage{amsmath,amssymb,amsfonts,amssymb,amsthm}
\newtheorem*{remark}{Remark}
\usepackage{setspace}
\usepackage{bbm}
\usepackage{float}
\usepackage{ulem}
\usepackage{booktabs}
\usepackage{hyperref}
\usepackage{xcolor}
\hypersetup{colorlinks,breaklinks,urlcolor=purple, linkcolor=cyan, citecolor = blue}
\usepackage{xcolor}

\begin{document}

\newcommand{\indep}{\rotatebox[origin=c]{90}{$\models$}}
\newcommand{\bb}{\boldsymbol{\beta}}
\newcommand{\bmu}{\boldsymbol{\mu}}
\newcommand{\by}{\boldsymbol{y}}
\newcommand{\bx}{\boldsymbol{x}}
\newcommand{\bu}{\boldsymbol{u}}
\newcommand{\ba}{\boldsymbol{\alpha}}
\newcommand{\bg}{\boldsymbol{\gamma}}
\newcommand{\bz}{\boldsymbol{z}}
\newcommand{\bZ}{\boldsymbol{Z}}
\newcommand{\bh}{\boldsymbol{h}}
\newcommand{\bd}{\boldsymbol{\delta}}
\newcommand{\bD}{\boldsymbol{D}}
\newcommand{\bt}{\boldsymbol{\theta}}

\def\spacingset#1{\renewcommand{\baselinestretch}%
{#1}\small\normalsize} \spacingset{1}

\if0\blind
{
  \title{\bf Inference with approximate local false discovery rates}
  \author{Rajesh Karmakar\thanks{Email: \href{mailto:rajeshkarmakar589@gmail.com}{\nolinkurl{rajeshkarmakar589@gmail.com}}}\hspace{.2cm}\\Department of Statistics and Operations Research, Tel-Aviv university\\
    \\
    Ruth Heller\thanks{Email: \href{mailto:ruheller@gmail.com}{\nolinkurl{ruheller@gmail.com}}}\hspace{.2cm}\\Department of Statistics and Operations Research, Tel-Aviv university\\
    \\
    Saharon Rosset\thanks{Email: \href{mailto:saharon@tauex.tau.ac.il}{\nolinkurl{saharon@tauex.tau.ac.il}}}\hspace{.2cm}\\Department of Statistics and Operations Research, Tel-Aviv university}
  \maketitle
} \fi


\bigskip
\begin{abstract}
Efron's two-group model is widely used in large scale multiple testing. This model assumes that test statistics are mutually independent, however  
in realistic settings they  are typically dependent, 
and taking the dependence into account can boost power. The general two-group model takes the dependence between the test statistics into account.  Optimal  policies in the general two-group model require calculation, for each hypothesis, of the probability that it is a true null given all test statistics, denoted  local false discovery rate (locFDR). Unfortunately, calculating  locFDRs under realistic dependence structures can be computationally prohibitive.  
We propose calculating approximate locFDRs based on a properly defined N-neighborhood for each hypothesis. We prove that by thresholding the approximate locFDRs with a fixed threshold, the marginal false discovery rate is controlled for any dependence structure. Furthermore, we prove that this is the optimal procedure in a restricted class of decision rules, where decision for each hypothesis is only guided by its N-neighborhood.  We show through extensive simulations that our proposed method achieves substantial power gains compared to alternative practical approaches, while maintaining conceptual simplicity and computational feasibility. We demonstrate the utility of our method on  a genome wide association study of height. 
\end{abstract}

\noindent%
{\it Keywords:}  dependent test statistics; large scale inference;  marginal false discovery rate (mFDR); multiple testing;  genome wide association studies (GWAS)
\vfill

\newpage
\spacingset{1.45} 
\section{Introduction}
\label{sec:intro}


The two-group model introduced in \cite{efron2001empirical}, is a popular mixture model-based procedure for large scale inference. Within this framework it is common to control false discovery by controlling the False Discovery Rate (FDR, \citealt{benjamini1995controlling}), or its closely related variant the marginal FDR (mFDR,  \citealt{genovese2002operating}). 

\cite{sun2007oracle} developed an adaptive multiple testing rule for mFDR control in the two-group model and showed that such a procedure is optimal in the sense that it minimizes the false non-discovery rate while controlling the mFDR under independence of tests. This rule is based on thresholding the marginal local FDR (locFDR), which is the conditional probability of the hypothesis being null, given only the observed test statistic for that hypothesis. If we denote the true state for test $i$ by $h_i \in \{0,1\}$ and the observed test statistics by $Z_i$, then marginal $locFDR = {\mathbb P}(h_i=0 | Z_i)$. In a series of papers \citep{efron2007correlation,efron2010correlated}, Efron has investigated the effect of correlations on the outcome of multiple testing procedures, 
being concerned that approaches that ignore dependence may not correctly control the FDR, so it is necessary to accommodate these correlations in multiple testing methodology. 

Putting power at the forefront, \cite{xie2011optimal} introduced the general two-group model with arbitrary dependence structure among z-scores. They showed that the optimal policy for mFDR control is thresholding the ``oracle'' locFDR statistics,  $\mathbb{P}(h_i = 0 \vert Z_1,\ldots,Z_K) = \mathbb{P}(h_i = 0 \vert \bZ)$,  with a fixed threshold. The oracle locFDR statistics explicitly use the correlation among the z-scores. 
Recently \cite{heller2021optimal} showed that the oracle locFDRs are the optimal test statistics for controlling FDR as well, under general dependence. They showed that for some special dependence structures like block dependence or equicorrelated setting, it may be practical to calculate these oracle locFDRs for dependent data. However beyond these special cases there is no computationally feasible method for calculating the full locFDR under dependence. A standard approach is to approximate the oracle locFDR statistics by sub optimal marginal locFDR statistics which are much easier to compute. By ignoring the correlation among tests statistics, this approach leads to substantial loss of power, as demonstrated empirically in our analyses below.

As a real life example, in GWAS we are interested in finding significant SNPs associated with a phenotype among millions of SNPs in the whole genome. 
It is well understood that nearby SNPs on the genome are highly correlated due to linkage disequilibrium (LD), while SNPs which are far apart are very weakly correlated. Marginal analysis is a standard procedure for finding associations in GWAS, but it does not identify the SNPs that are associated with the phenotype given all other SNPs. For a joint analysis, one approach is to approximate the oracle locFDRs using 
a complete Bayesian treatment, which relies on methods like Markov Chain Monte Carlo (MCMC) for inference. For example, \cite{zhu2017bayesian} developed a Bayesian approach for finding association using GWAS summary statistics which takes into account the correlations among the SNPs. Their approach relies on publicly available estimates of effect sizes for $K$ SNPs $\hat{\beta}_j$ for $1\le j \le K$, calculated from marginal regressions without taking correlations into account, and builds a complete model taking into account the correlations between these estimates, using a known correlation (LD) matrix. They then implement a full Bayesian approach using MCMC to ultimately infer the posterior null probabilities $\mathbb{P}(\beta_j = 0 \vert \hat{\bb})$ where $\hat{\bb} = (\hat{\beta_1},\ldots,\hat{\beta_K})$. In practice, their procedure is very time consuming, as each MCMC iteration takes substantial computation, and the entire procedure has to be repeated several times to confirm convergence of the MCMC solution.

In this paper we propose a novel approach for mFDR control under dependence in the two-group model, which can be viewed as a compromise between the simplistic approach of relying on marginal locFDRs and the computationally prohibitive approach of calculating the full oracle locFDRs and obtaining optimal power. We propose to define the {\em N-neighborhood} for each test, for some small $N\in \mathbb N,$ and only consider multiple testing procedures that  base the decision for test $i$ only on its $N$ neighborhood. We enote this class of procedures (or decision policies) by $C_N.$ For example, in our motivating application to GWAS, LD structure implies physically close SNPs are correlated, and the range of correlation is roughly similar throughout the genome,  so the neighborhood for each SNP is defined by the group of $2\times N$ SNPs surrounding it.

The rest of the paper is organized as follows.

    In \S~\ref{sec:NandG} we  define all the necessary notations and explain our goal in this paper.
In \S~\ref{sec:meth} we prove in Theorem \ref{theorem:theorem1} that within the reduced class $C_N$, the optimal test that maximizes power while controlling mFDR in the two-group model, relies on thresholding the posterior probabilities $T_{i,N}=\mathbb{P}(h_i=0\vert Z_{i-N},\dots,Z_i,\dots,Z_{i+N}),$ which we term $locFDR_N.$We further show in Corollary \ref{cor1}  that the power (expected number of true discoveries) is an increasing function of $N$ for the resulting rules, making explicit the  trade off between power and computation that the neighborhood approximation provides.
    In  \S~\ref{sec:CompEst} we show that the computation of the $locFDR_N$ statistics is linear in the number of hypothesis. Furthermore we show how we can estimate the $locFDR_N$ statistics assuming a reasonable model for GWAS. Finally we summarize our complete data driven algorithm. 
    In  \S~\ref{sec:verify} we perform extensive simulations with different covariance structures and we observe that the $locFDR_N$ statistics achieves more power compared to other 
    practical methods. In \S~\ref{sec:UKB} we provide a complete analysis pipeline for GWAS using summary statistics. Required input data are marginal SNP effect sizes ($\hat\beta_j$), standard error of the effect (s.e.($\hat\beta_j$) $\approx \hat \sigma_j$) and estimated matrix $\hat R$ of LD (correlations) among Z-scores ($Z_j = {\hat\beta_j}/{\hat \sigma_j}$) of SNPs.
\section{Notation and Goal}\label{sec:NandG}
We assume the data are generated from the generalized two group model. The hypothesis states vector, $\bh=(h_1,\dots,h_K)$, has entries sampled independently from the $Ber(\pi)$ distribution. Given the hypothesis states $\bh$, the sequence of test statistic for testing the $K$ null hypothesis, $h_i = 0$, $i=1,\ldots, K$,  follows $\bZ = (Z_1,\dots,Z_K)  \sim g(\bz \vert \bh)$. Based on $\bZ$, the goal is to construct a multiple testing procedure, defined by the decision rule $\bD(\bZ) = (D_1(\bZ),\dots,D_K(\bZ)):\mathbb{R}^K \to \{0,1\}^K$,  which has good power while controlling a meaningful error.

Let $T_i(\bZ)$ denote the oracle locFDR, which is the statistic for optimal decisions \citep{xie2011optimal, heller2021optimal}: 
$$\quad T_i(\bz) = \mathbb{P}(h_i =0 \vert \bZ = \bz) = \dfrac{\sum_{\bh, h_i =0}{\mathbb{P}(\bz \vert \bh) \mathbb{P}(\bh)}}{\sum_{\bh}{\mathbb{P}(\bz \vert \bh) \mathbb{P}(\bh)}}.$$

The expected number of the true  and false positives are simple expressions of the oracle locFDRs: \begin{eqnarray} 
&& TP(\bD) =\mathbb{E}_{\bZ,\bh}(\bh^T \bD) = \mathbb{E}_{\bZ}\left(\sum_{i=1}^{K}{D_i(\bZ)(1-T_i(\bZ))}\right)\nonumber \\
&& \mathbb{E}(V(\bD)) =\mathbb{E}_{\bZ,\bh}((1-\bh)^T \bD) = \mathbb{E}_{\bZ}\left(\sum_{i=1}^{K}{D_i(\bZ)T_i(\bZ)}\right) \nonumber
\end{eqnarray}

A relevant error rate is the marginal false discovery rate (mFDR), $$mFDR(\bD) = \dfrac{\mathbb{E}_{\bZ,\bh}(V(\bD))}{\mathbb{E}_{\bZ,\bh}(R(\bD))},$$
where $\quad \mathbb{E}_{\bZ,\bh}(R(\bD)) = \mathbb{E}_{\bZ,\bh}(V(\bD)) + TP(\bD)$ is the expected number of rejections.

Note that some expectations above are over the joint distribution of $(\bZ,\bh)$ and some are over $\bZ$. From now on we will not denote the random variable with respect to which we take the expectation since it is easily understandable from the context.
It has been shown in \cite{xie2011optimal} that the optimal rule over the class of all decision functions $\bD:\mathbb{R}^K \to \{0,1\}^K$ for the above problem is thresholding the locFDR statistic with a fixed threshold. So the decision rule is 
    $D_i(\bz) = \mathbb{I}(T_i(\bz)\le c)$, where $c$ is a  fixed constant depending on $\alpha$.
The complexity required for naively calculating all the locFDR statistics is $O(K2^K)$, which is typically infeasible. Hence the optimal procedure is impractical, except for some very special types of dependencies across the $z$-scores, e.g.,  block dependence with small block sizes \citep{heller2021optimal}.
Due to the difficulty in finding the globally optimal solution, we suggest considering  a smaller class of decision functions  which depend only on $N$ neighbors, for $N \in \{0,\ldots, K\}$:
\begin{gather}
    \mathcal{C}_N = \left\{\bD:\mathbb{R}^K\to \{0,1\}^K \quad \vert \quad \text{$\bD$ is of the form} \quad \bD(\bz) = (D_1(\bz_{1,N}),\dots,D_K(\bz_{K,N})) \right\},\\
    \text{where} \quad \bz_{i,N} = (z_{\max(i-N,1)}, \dots, z_{i-1},z_i,z_{i+1}, \dots , z_{\min(i+N,K)}), \quad 1\le i\le K.\nonumber
\end{gather}
In this definition we assume that the tests are ordered and each test's "neighbors" are defined by the ones next to it in the order. We also assume that the neighborhood size $N$ is fixed for all tests. Both of these assumptions can be relaxed, but we make them here for simplicity of notation and presentation of the resulting algorithms. 

Given a choice of $N$ and the resulting class $C_N,$ we  show in \S~\ref{sec:meth} that the optimal testing procedure,  which is limited to rules in $C_N$, relies  on the localFDR after marginalization so that the test statistic is within this class: 
\begin{equation}
    T_{i,N} = \mathbb{P}(h_i=0 \vert \bZ_{i,N}).
\end{equation}
Henceforth, we refer to $T_{i,N}$ as the $locFDR_N$. For $N=0$, $C_0$ includes only marginal rules and hence $locFDR_0$ is the marginal locFDR \citep{efron2001empirical}, while for $N=K$, $locFDR_K$ \citep{sun2009large,xie2011optimal,heller2021optimal} is the true locFDR. We can view increasing the size of $N$ as offering a trade-off between computation and power: for small $N$, calculating $locFDR_N$ and the resulting optimal rules within the limited class $C_N$ is computationally efficient but can suffer loss of power, and as $N$ increases computations become exponentially more complex but power improves. 
As we demonstrate below, in GWAS-like situations with a natural neighborhood structure and  strong dependence only within neighborhood, the power increase can be rapid when $N$ is small, and we can obtain excellent power-computation tradeoff (Figure~\ref{fig:sub1}). For other types of correlation structures we also get a nice trade-off (Figure~\ref{fig:sub2}, \ref{fig:sub3}) . Here our aim is not to find the optimal solution (which is not implementable), but a relaxation of the optimal solution which improves power significantly compared to existing marginal test statistics. We show in \S~\ref{sec:verify} that the power gain using the $locFDR_N$ statistic, compared to existing methods, can be substantial. 

\section{OMT Policy in the restricted class $\mathcal{C}_N$}
\label{sec:meth}

In the following proposition we state a simple but useful result.
\begin{prop}
\label{prop:prop1}
Define the quantity

\begin{equation} \label{main_rule}
    t_{\alpha, N} = \sup \left\{t:\dfrac{\mathbb{E}\sum_{i = 1}^{K}{\mathbb{I}(T_{i,N}\le t)T_{i,N}}}{\mathbb{E}\sum_{i = 1}^{K}{\mathbb{I}(T_{i,N}\le t)}}\le \alpha\right\}
\end{equation}
for $N \ge 0$ and $0 < \alpha < 1$. Then the procedure which rejects $H_0^i$ if $T_{i,N} \le t_{\alpha, N}$ for $i = 1,\dots, K$, controls the mFDR at level $\alpha$.
\end{prop}
\begin{proof}
The mFDR of the procedure which rejects $H_0^i$ if $T_{i,N} \le t_{\alpha,N}$ is given by
\begin{align*}
    \text{mFDR} \quad &=  \dfrac{\mathbb{E}\left[\sum_{i = 1}^{K}{\mathbb{I}(T_{i,N}\le t_{\alpha,N},h_i = 0)}\right]}{\mathbb{E}\left[\sum_{i = 1}^{K}{\mathbb{I}(T_{i,N}\le t_{\alpha,N})}\right]} = \dfrac{\sum_{i = 1}^{K}{\mathbb{E}_{(Z_{i,N},h_i)}\left[\mathbb{I}(T_{i,N}\le t_{\alpha,N},h_i = 0)\right]}}{\sum_{i = 1}^{K}{\mathbb{E}_{(Z_{i,N})}\left[\mathbb{I}(T_{i,N}\le t_{\alpha,N})\right]}} \\
    &= \dfrac{\sum_{i = 1}^{K}{\mathbb{E}_{Z_{i,N}}\mathbb{E}_{h_i \vert Z_{i,N}}\left[\mathbb{I}(T_{i,N}\le t_{\alpha,N},h_i = 0)\right]}}{\sum_{i = 1}^{K}{\mathbb{E}_{(Z_{i,N})}\left[\mathbb{I}(T_{i,N}\le t_{\alpha,N})\right]}} =   \dfrac{\sum_{i = 1}^{K}{\mathbb{E}_{Z_{i,N}}\left[\mathbb{I}(T_{i,N}\le t_{\alpha,N})\mathbb{P}(h_i = 0\vert Z_{i,N})\right]}}{\sum_{i = 1}^{K}{\mathbb{E}_{(Z_{i,N})}\left[\mathbb{I}(T_{i,N}\le t_{\alpha,N})\right]}} \\
    &= \dfrac{\sum_{i = 1}^{K}{\mathbb{E}\left[\mathbb{I}(T_{i,N}\le t_{\alpha,N})T_{i,N}\right]}}{\sum_{i = 1}^{K}{\mathbb{E}\left[\mathbb{I}(T_{i,N}\le t_{\alpha,N})\right]}} \le \alpha.
\end{align*}
where the  last inequality follows from the definition of $t_{\alpha,N}$.
\end{proof}
A similar result was given in \cite{heller2021optimal} for the marginal locFDR statistics $\{T_{1,0},\dots,T_{K,0}\}$, but we expect more power by using the procedure based on $T_{i,N}$, $N>0$.

The optimal multiple testing procedure with mFDR control is a solution to the problem:
\begin{equation}\label{complete_rule}
    \max_{\bD:\mathbb{R}^K\to \{0,1\}^K}{TP(\bD)} \quad \text{subject to} \quad mFDR(\bD)\le\alpha.
\end{equation}

Now we concentrate on the class $\mathcal{C}_N$ and find an optimal solution within this restricted class, so the optimization problem is 
\begin{equation}
\label{relaxed_rule}
            \max_{\bD \in \mathcal{C}_N}{TP(\bD)} \quad \text{subject to} \quad mFDR(\bD)\le\alpha.
\end{equation}

\begin{theorem}
\label{theorem:theorem1}
The Procedure detailed in Proposition \ref{prop:prop1} is the 
solution to the OMT problem in Equation~\eqref{relaxed_rule}, i.e., the optimal policy is 
$\bD_N^* = \left(\mathbb{I}(T_{1,N}\le t_{\alpha,N}),\dots,\mathbb{I}(T_{K,N}\le t_{\alpha,N})\right)$
\end{theorem}
\begin{proof}
see Appendix~\ref{appendix:A}.
\end{proof}

\begin{remark}
Though marginal test statistics ($locFDR_0$) based fixed threshold rules were already present long before, the optimality of such rules for dependent test statistics was not known (as far as we know). Therefore,  the above result is not only new for $0<N<K$ but also for $N=0$.
\end{remark}

\begin{corollary}\label{cor1}
The optimal fixed threshold mFDR level $\alpha$ rule based on $T_{i,N+1}$ has more power than the optimal fixed threshold mFDR level $\alpha$ rule based on $T_{i,N}$ for every $N\ge 0$.
\end{corollary}

\begin{proof}
Theorem~\ref{theorem:theorem1} says that the optimal fixed threshold mFDR level $\alpha$ rule based on $T_{i,N}$ is the solution to problem \eqref{relaxed_rule}. But it is obvious that $\mathcal{C}_N \subset \mathcal{C}_{N+1}$ for every $N\ge 0$. Hence the proof.
\end{proof}

\section{A data driven procedure 
}\label{sec:CompEst}
\subsection{The computational complexity of the $locFDR_N$ statistics}
To implement the procedure described in Proposition~\ref{prop:prop1}, we need to compute the statistics $\{T_{1,N},\dots,T_{K,N}\}$. We denote $l = \min(i+N,K) - \max(i-N,1) + 1$. Clearly $l \le 2N + 1$, so  $T_{i,N}$ can be written as
\begin{align} \label{nbd_locfdr}
    T_{i,N} & = \dfrac{\sum_{\bh_{i,N}, h_i =0}{\mathbb{P}(\bZ_{i,N} \vert \bh_{i,N}) \mathbb{P}(\bh_{i,N})}}{\sum_{\bh_{i,N}}{\mathbb{P}(\bZ_{i,N} \vert \bh_{i,N}) \mathbb{P}(\bh_{i,N})}} 
\end{align}
where $\mathbb{P}(\bh_{i,N})$ is the probability distribution of $\bh_{i,N}$ and $\mathbb{P}(\bZ_{i,N}\vert\bh_{i,N})$ is the conditional density of $\bZ_{i,N}$ given $\bh_{i,N}$. The sum in the denominator above is over $l$ dimensional binary vectors and hence requires $2^l$ evaluations of $\mathbb{P}(\bZ_{i,N} \vert \bh_{i,N})$ and each of these takes time of at most $O(l^2)$. Hence calculating the denominator requires time $O(2^l l^2)$. Note that all the calculations in the denominator can be stored with $O(2^l)$ memory and hence to evaluate the numerator, we do not need any extra computation time. Since $O(2^l l^2) \le O(2^{2N+1}(2N + 1)^2)$,  the worst case runtime for computing all the $locFDR_N$s is $O(K2^{2N+1}(2N + 1)^2)$. 

Note that for finding the globally optimal rule, we had to calculate the oracle $locFDR$s $T_1,\dots,T_K$. If $g$ has no special structure, calculating these quantities requires runtime of $O(K2^K)$. 
\subsection{Estimation strategies}\label{estimation_strategies}
Although computations of the statistics of interest $\{T_{1,N},\dots,T_{K,N}\}$ are simple with known probabilistic structure, in real life examples we need to estimate the process parameters from the data. Below we discuss how to estimate these quantities in situation where the dependence can be assumed to be of short range, which approximates many real life applications. 

Assume we have the following model
\begin{equation}{\label{model}}
    \begin{gathered}
    \bZ\vert \bh \sim \mathcal{N}_K(b\boldsymbol{\bh},\Sigma+\tau^2 diag(\bh)) \\
    h_i \overset{i.i.d.}{\sim} Ber(\pi), \quad i=1,\ldots,K,
\end{gathered}
\end{equation}
where $\Sigma$ is known, with diagonal entries of one. 
This model assumes that a null (i.e., with $h_i = 0$) SNP's  test statistic has a $\mathcal{N}(0,1)$ distribution, and a nonnull (i.e., with $h_i = 1$)  SNP's test statistic has a   $\mathcal{N}(b,1+\tau^2)$ distribution. It is reasonable, e.g.,  for testing the correlations between covariates and outcome in a linear model, where $\Sigma$ is determined by the design matrix of the covariates,  see \S~\ref{sec:UKB} for details. 
This model is similar to the ``RSS-BVSR" model considered in \cite{zhu2017bayesian}, and it reduces to the two-group model of \cite{efron2001empirical} if $\Sigma$ is a diagonal matrix. 

In practical applications, we want to estimate the model \eqref{model} from the observed summary statistics $\bZ$.  If  we can assume the dependence to be of short ranged (as in GWAS), then we can find a subset of $z$-scores. which are approximately independent. Due to LD, nearby SNPs are correlated and SNPs which are far apart tend to be independent of each others. As a result, if we select the z-scores of the SNPs based on their position on the chromosome, we can get approximately independent collection of $\Tilde{K}$ $z$-scores. After we select these $z$-scores forestimation, they can be assumed to follow the following simpler model two group model: 
\begin{equation}\label{eq-mixturemodel}
\begin{gathered}
    h_i \overset{i.i.d.}{\sim} Ber(\pi)\\
    z_i \vert h_i \overset{ind}{\sim} \mathcal{N}(bh_i,1+\tau^2h_i)
\end{gathered}
\end{equation} It is  easier to find the estimates of $\pi,b, \tau$ from the model \eqref{eq-mixturemodel} using the subset of $z$-scores which we call $\bz_s$. Once these parameters are estimated, we plug  these estimates into model \eqref{model}. Specifics follow.

Our first step is thus the selection of the sub-vector $\bz_s$. We do not assume the i.i.d. two group model on $\bz$, but we assume that  most of the time we are able to extract a subset of z-scores which either exactly or approximately follow a i.i.d. two group model due to the assumption of short range strong dependency. Details of how estimation is done in case Banded(1), AR(1) covariance are given in section~\ref{subsec: Data_Driven}.

Next, we need to estimate $\pi, b, \tau^2$ using $\bz_s$. We apply the EM algorithm \citep{dempster1977maximum} to estimate the parameters as suggested in \cite{peel2000finite}. We consider the following two methods, which differ in the way they handle the estimation of $\pi$. 

\textbf{Estimating Parameters 1:} The updates of the parameters are given below
\begin{equation}\label{equation:est_fullem}
\begin{gathered} 
    \pi^{(t+1)} = \dfrac{1}{\Tilde{K}}\sum_{i}{\mathbb{P}(h_i = 1 \vert z_i,\boldsymbol{\theta}^{(t)})}\\
    b^{(t+1)} = \dfrac{\sum_{i}\mathbb{P}(h_i = 1 \vert z_i,\boldsymbol{\theta}^{(t)})z_i}{\sum_{i}\mathbb{P}(h_i = 1 \vert z_i,\boldsymbol{\theta}^{(t)})}\\
    {\tau^2}^{(t+1)} =\max\left\{0, \dfrac{\sum_{i}\mathbb{P}(h_i = 1 \vert z_i,\boldsymbol{\theta}^{(t)})(z_i-b^{(t+1)})^2}{\sum_{i}\mathbb{P}(h_i = 1 \vert z_i,\boldsymbol{\theta}^{(t)})}-1\right\}
\end{gathered}
\end{equation}

where $\boldsymbol{\theta}^{(t)} = (\pi^{(t)},b^{(t)},{\tau^2}^{(t)})$ are the estimates of the parameters in the $t$-th step and $\bz_s = (z_1,\dots z_{\Tilde{K}})$. The derivation of the equations~\eqref{equation:est_fullem} is detailed in appendix~\ref{appendix:B}. 

\textbf{Estimating Parameters 2:}  Alternatively we  estimate $\pi$ from $\bz$ using the well known method of \cite{jin2007estimating} and plug-in this estimate of $\pi$ in the EM-iterations. Let $\hat{\pi}$ be the estimate from $z$-scores $\bz$ of non-null proportion $\pi$. Then the updates of the other parameters are given below
\begin{equation}\label{equation:estimate_partialem}
\begin{gathered}
    b^{(t+1)} = \dfrac{\sum_{i}\mathbb{P}(h_i = 1 \vert z_i,\boldsymbol{\psi}^{(t)})z_i}{\sum_{i}\mathbb{P}(h_i = 1 \vert z_i,\boldsymbol{\psi}^{(t)})} \\
    {\tau^2}^{(t+1)} =\max\left\{0, \dfrac{\sum_{i}\mathbb{P}(h_i = 1 \vert z_i,\boldsymbol{\psi}^{(t)})(z_i-b^{(t+1)})^2}{\sum_{i}\mathbb{P}(h_i = 1 \vert z_i,\boldsymbol{\psi}^{(t)})}-1\right\}
\end{gathered}
\end{equation}
where $\boldsymbol{\psi}^{(t)} = (\hat{\pi},b^{(t)},{\tau^2}^{(t)})$ are the estimates of the parameters in the $t$-th step. The derivation of the the equations~\eqref{equation:estimate_partialem} follows similarly as done in appendix~\ref{appendix:B}.

 Using the estimated parameters $\hat b, \hat \pi, \hat \tau$, we estimate the $K$  $locFDR_N$,  and denote  these statistics by $\{\hat T_{1,N},\dots,\hat T_{K,N}\}$. 
For inference using the estimated $locFDR_N$s,  we need to determine the rejection cutoff, so that  hypotheses with estimated $locFDR_N$s below that cutoff will be rejected. We determine this cutoff numerically as follow. We draw $B$ different data sets of $z$-scores from the mixture model in \eqref{eq-mixturemodel} with estimated parameters $\hat b, \hat \pi, \hat \tau$. We estimate $Q(t) = \dfrac{\mathbb{E}\sum_{i = 1}^{K}{\mathbb{I}(\hat T_{i,N}\le t)\hat T_{i,N}}}{\mathbb{E}\sum_{i = 1}^{K}{\mathbb{I}(\hat T_{i,N}\le t)}}$ by $\hat{Q}(t) = \dfrac{\sum_{b=1}^{B}\sum_{i = 1}^{K}{\mathbb{I}(\hat T^b_{i,N}\le t)\hat T^b_{i,N}}}{\sum_{b=1}^{B}\sum_{i = 1}^{K}{\mathbb{I}(\hat T^b_{i,N}\le t)}}$ where $\{\hat T^b_{1,N},\dots,\hat T^b_{K,N}\}$ has been calculated from the $b$-th $z$-vector and estimated parameters $\hat b, \hat \pi, \hat \tau$. Now we perform a simple line search to find  \begin{equation}\label{eq-estimated cutoff}\hat{t}_{\alpha,N}= \sup \left\{t: \hat{Q}(t) \le \alpha\right\}.\end{equation}

We summarize the whole pipeline concisely  in Algorithm \ref{alg:algo1}. 

\begin{algorithm}
\label{alg:algo1}
{\fontsize{11pt}{11pt}\selectfont}
\begin{algorithmic}[1]

  \scriptsize
  \footnotesize{\STATE {Step 1 - \textbf{Input:} The vector of $z$-scores $\bz$, covariance among the $z$-scores $\Sigma$, realistic $N \ge 0$, $B > 0$.}
  \STATE {Step 2 - Choose a sub vector $\bz_s$ of approximately independent $z$-scores from the whole vector of $z$-scores $z$ as described in pragraph \textbf{Selecting sub-vector}} in \S~\ref{estimation_strategies}.
  \STATE {Step 3 - Given $\bz_s$ and $\bz$, estimate $b, \pi, \tau$ using either paragraph \textbf{Estimating Parameters 1} or paragraph \textbf{Estimating Parameters 2} as described in \S~\ref{estimation_strategies} and call these estimates $\hat{b}, \hat{\pi}, \hat{\tau}$.}
  \STATE {Step 4 - Calculate $\{T_{1,N},\dots,T_{K,N}\}$ using the estimated parameters in step 3 and the $z$-scores by the formula \eqref{nbd_locfdr} and call these estimated statistics $\{\hat T_{1,N},\dots,\hat T_{K,N}\}$}.
  \STATE {Step 5 - Calculate the estimated cutoff $\hat{t}_{\alpha,N}$ in \eqref{eq-estimated cutoff}}. 
  \STATE { Step 6 - \textbf{Output:} The set $\mathcal{R}_\alpha = \{i: \hat T_{i,N}\le \hat{t}_{\alpha,N}$\} }} of rejected hypothesis.
\end{algorithmic}
\caption{The Data-driven $locFDR_N$ procedure for level $\alpha$ mFDR control}
\label{alg:algo1}
\end{algorithm}

\section{Simulations}
\label{sec:verify}
\subsection{Simulation from Mixture of Gaussians (Oracle)}
Here we show all the results with known values of the parameters of the distribution described by \eqref{model}.

In the rest of the paper, we will consider covariance structures for $\Sigma=(\sigma_{i,j})_{1\le i, j\le K}$ among the following four: 
\begin{itemize}
    \item \textbf{AR(1):} 
    \begin{equation}\label{eq:AR(1)}
    \sigma_{i,j} = \rho^{|i-j|}
    \end{equation} 
    for $1\le i, j\le K$ where $-1 \le \rho \le 1$. 
    \item \textbf{Banded(1):} \begin{equation}\label{eq:Banded(1)}
    \sigma_{i,j} = 
\begin{cases}
    1,& \text{if } i=j\\
    \rho,& \text{if } |i-j|=1\\
    0,& \text{if } |i-j|>1
\end{cases}
\end{equation}
    for $1\le i, j\le K$ where $-0.5 \le \rho \le 0.5$. 
    \item \textbf{Long Range:} \citep{fan2017estimation} 
    \begin{equation}\label{eq:FGN}
    \sigma_{ij} = \frac{1}{2}\left(\vert\vert i-j\vert+1\vert^{2H}-2\vert i-j \vert^{2H}+\vert\vert i-j\vert -1 \vert^{2H}\right)\end{equation}
    $1\le i, j\le K$, where $\frac{1}{2}< H < 1$ with $H \to \frac{1}{2}$ indicating weak dependence and $H \to 1$ indicating strong long range dependence.
    \item \textbf{Equicorrelated:} \begin{equation}\label{eq:equi}
    \sigma_{i,j} = \begin{cases}
    1,& \text{if } i=j\\
    \rho,& \text{if } |i-j|\ne 0
\end{cases}
    \end{equation}
    for $1\le i, j\le K$ where $-1 \le \rho \le 1$. 
\end{itemize}

Figure \ref{fig:fig1} shows a scatter plot of $T_{i,N}$ versus $T_{i,N+1}$, and the marginal density of $T_{i,N}$, for $N=0,\ldots,5$. 
As we move to comparison between higher order locFDRs, the simulated locFDRs tend to concentrate around the 45 degree line. The concentration is particularly good for AR(1) and the long-range correlation structure. The density plots show that for all dependencies, the densities are similar for $N>0$ and  strikingly different than  the density for $N=0$. 
\begin{figure} \label{fig:fig1}
\begin{subfigure}{\textwidth}
  \centering
  \includegraphics[width=0.8\linewidth]{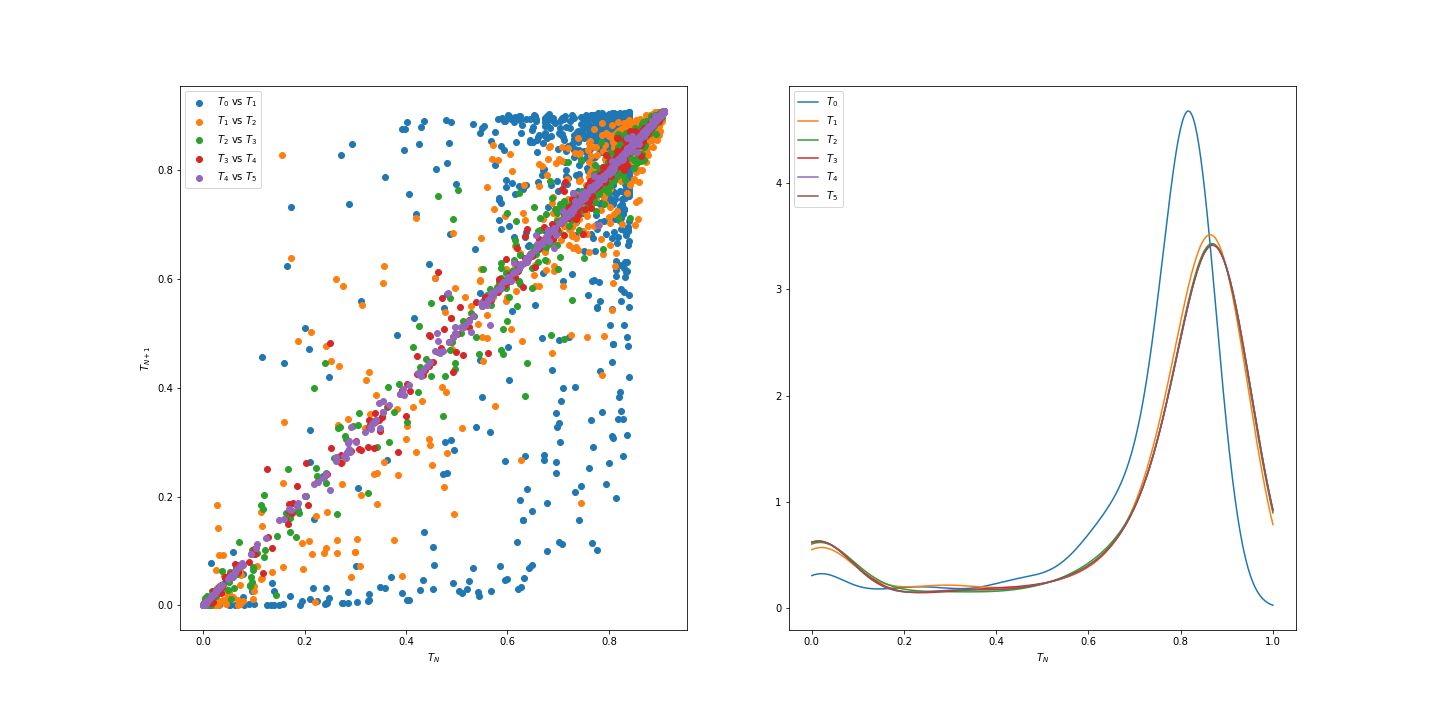}  
  \caption{AR(1) covariance with $\rho = 0.8$ as described in equation~\eqref{eq:AR(1)}}
  \label{fig:sub1}
\end{subfigure}


\begin{subfigure}{\textwidth}
  \centering
  \includegraphics[width=0.8\linewidth]{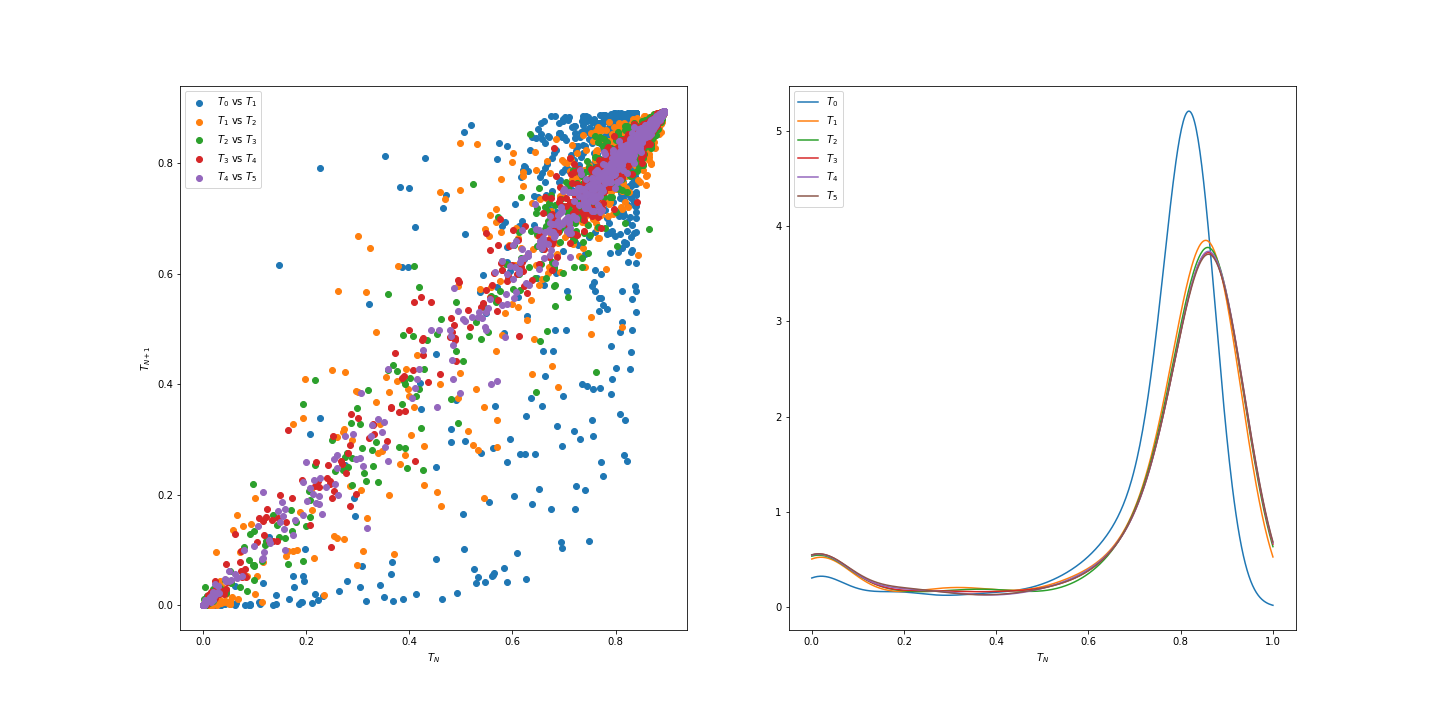}  
  \caption{Long-range covariance with $H = 0.9$ as described in equation \eqref{eq:FGN}}
  \label{fig:sub2}
\end{subfigure}

\begin{subfigure}{\textwidth}
  \centering
  \includegraphics[width=0.8\linewidth]{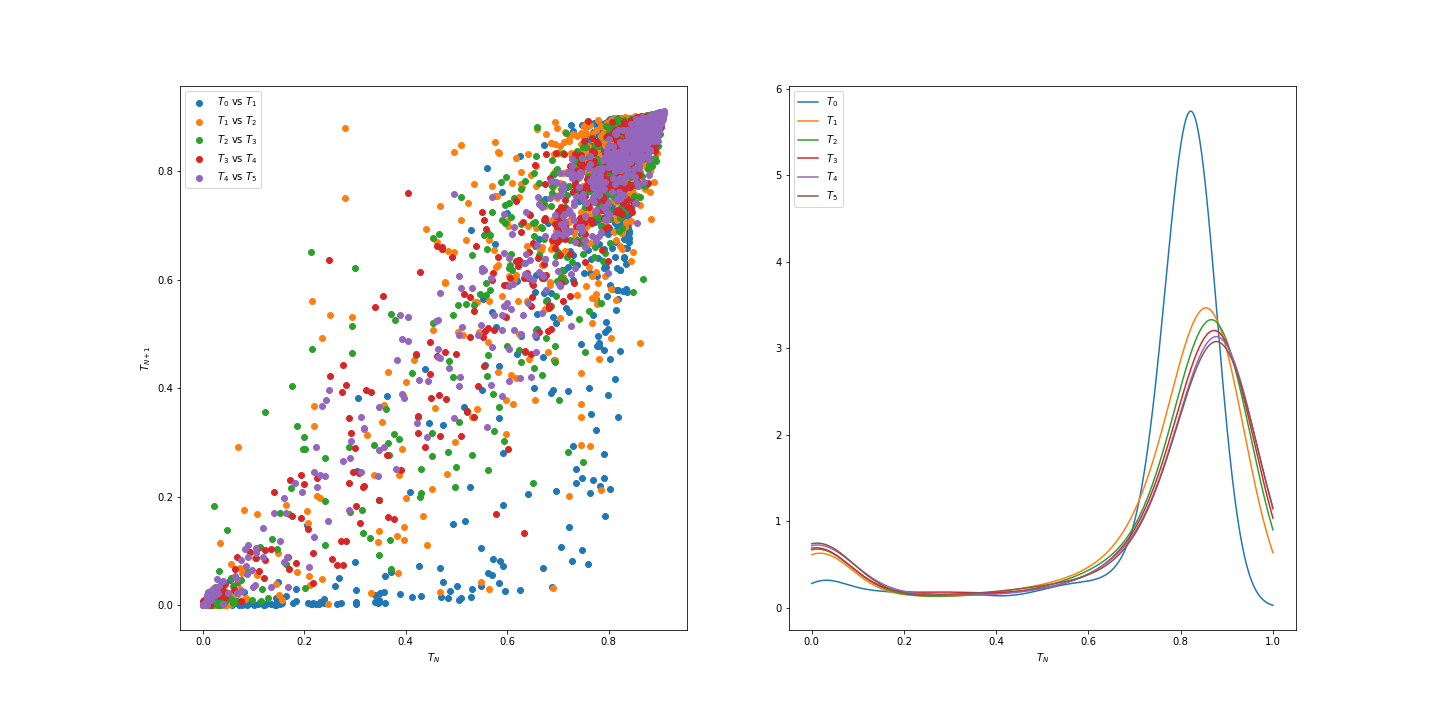}  
  \caption{Equicorrelated covariance with $\rho = 0.8$ as described in equation \eqref{eq:equi}}
  \label{fig:sub3}
\end{subfigure}
\caption{The scatter plots ($T_{i,N}$, $T_{i,N+1}$) (left column) and the density of $T_{i,N}$ (right column), for $N=0,\ldots,5$.  Rows differ  by the dependence structure, $\Sigma$, among the $z$-scores. The data generation of the $z$-score vector is  model~\eqref{model}, with $K=1000, \tau=2,\pi=0.3$, and $b=0.2$.}
\label{fig:fig1}
\end{figure}

We make the following observations from the results in Table~\ref{table:table1}:
\begin{itemize}
    \item \textbf{Power}: as expected,  the power increases with $N$.
    \item \textbf{Effect of Correlation}: What makes the $locFDR_N$ most interesting is the rate of power increase, as a function of $N$: the power increase is  large when moving from  $N=0$ to $N=1$, and it is much smaller when moving from $N=1$ to $N=2$. This behaviour is common observed under all dependency structures considered. 
    The highest power increase was observed for the equicorrelated covariance structure. 
    \item \textbf{Variability of the FDP}: When there are strong correlations, FDP is highly variable. In all settings, the procedure with marginal statistics shows  higher variability in FDP compared to procedures that use $locFDR_N$ statistics with $N>0$. We consistently observe from simulations that 
 the variability of the FDP is decreasing in $N$, so this is an added advantage of the use of $T_{i,N}$ with $N$ as large as (computationally) possible.
    \item \textbf{mFDR vs FDR}:  mFDR and FDR differ when there is strong correlations using $locFDR_0$ statistics. But when we use the $locFDR_2$ statistics, more discoveries are made and the difference between  mFDR and FDR is small, regardless of the correlation structure. 
    \item \textbf{Cutoff}: All the simulations shows that cutoffs $t_{\alpha,N}$ for fixed $\alpha$ are increasing in $N$,  and  the change from $N=1$ to $N=2$ is much smaller than the change from $N=0$ to $N=1$. The decrease of the sequence $t_{\alpha,N+1}-t_{\alpha,N}$ depends on the covariance structure,  and the rate is higher in the short range setting, than in  the long range or equicorrelated setting. 
\end{itemize}


\vspace{5mm}
\hspace{0mm}
\begin{table}\label{table:table1}
\begin{tabular}{|l|l|c|c|c|c|c|c|}
\hline
\multicolumn{8}{|c|}{Fixed Parameters in all Simulations, $K=1000$, $\pi = 0.3$} \\
\hline
\multicolumn{1}{|c|}{Correlation}& \multicolumn{1}{|c|}{Method} & \multicolumn{2}{c|}{$\boldsymbol{T}_{0}$-rule} & \multicolumn{2}{c|}{$\boldsymbol{T}_{1}$-rule} & \multicolumn{2}{c|}{$\boldsymbol{T}_{2}$-rule}  \\
\cline{2-8}
 & \multicolumn{1}{|c|}{Error} & Estimate & s.e. & Estimate & s.e. & Estimate & s.e.\\
\hline

AR(1) &  mFDR & 0.0501 & 0.0015 & 0.0499& 0.0009 & 0.0508 & 0.0008 \\
$\rho=0.8$ &  FDR & 0.0494 & 0.0014 & 0.0498 & 0.0009 & 0.0507 & 0.0008 \\
$b=0$ &  TP & 61.19  & 0.352 & 124.08 & 0.428 & 139.364 & 0.448\\
$\tau = 2$ & Cutoff (t) & 0.1742 & - & 0.2356 & - & 0.2729 & - \\
\hline

Long Range & mFDR & 0.0489 & 0.0015 & 0.0506 & 0.001 & 0.0497 & 0.0010 \\
$H=0.8$ & FDR & 0.0481 & 0.0014 & 0.0506 & 0.001 & 0.0497 & 0.0010 \\
$b=0$ & TP & 61.95 & 0.3413 & 84.256 & 0.3815 & 88.37 & 0.3944\\
$\tau = 2$ & Cutoff (t) & 0.1742 & - & 0.2014 & - & 0.2074 & - \\
\hline

Equi & mFDR & 0.0731 & 0.0249 & 0.0515 & 0.0032 & 0.0501 & 0.0012 \\
$\rho=0.8$ & FDR & 0.0149 & 0.0035 & 0.0489 & 0.0023 & 0.0497 & 0.0012 \\
$b=0$ & TP & 62.026  & 0.8173 & 121.374 & 0.4404 & 146.29 & 0.4835\\
$\tau = 2$ & Cutoff (t) & 0.1742 & - & 0.2361 & - & 0.2903 & - \\
\hline

AR(1) & mFDR & 0.0510 & 0.0012 & 0.0487 & 0.001 & 0.0487 & 0.001 \\
$\rho=0.5$ & FDR & 0.0509 & 0.0012 & 0.0488 & 0.001 & 0.0487 & 0.0010 \\
$b=0$ & TP & 62.112  & 0.3465 & 83.85 & 0.3989 & 85.824 & 0.4019 \\
$\tau = 2$ & Cutoff (t) & 0.1742 & - & 0.1989 & - & 0.2029 & - \\
\hline

Long Range & mFDR & 0.0506 & 0.0013 & 0.0505 & 0.0011 & 0.0513 & 0.0011 \\
$H=0.7$ & FDR & 0.0504 & 0.0013 & 0.0507 & 0.0011 & 0.0514 & 0.0011 \\
$b=0$ & TP & 61.922 & 0.3503 & 70.282 & 0.3748 & 71.504 & 0.3786\\
$\tau = 2$ & Cutoff (t) & 0.1742 & - & 0.1823 & - & 0.1842 & - \\
\hline

Equi & mFDR & 0.0452 & 0.0049 & 0.0502 & 0.0023 & 0.0496 & 0.0014 \\
$\rho=0.5$ & FDR & 0.0304 & 0.0028 & 0.0473 & 0.0021 & 0.0491 & 0.0014 \\
$b=0$& TP & 61.848  & 0.5322 & 81.152 & 0.4022 & 91.166 & 0.4055\\
$\tau = 2$ & Cutoff (t) & 0.1742 & - & 0.1971 & - & 0.2088 & - \\
\hline
\end{tabular}
\caption{Power Gain with different covariance structure (Known Parameters)\\
For different parameters of the model~\eqref{model} (rows), using the procedure with $T_{i,0}$ (columns 3-4), $T_{i,1}$ (columns 5-6), and $T_{i,2}$ (columns 7-8): 
the ``Estimate" column indicates the estimates of different error rates and the ``s.e." column indicates the the standard error of these estimates. Standard Errors for mFDR were determined by bootstrapping, based on 500 data generations. }
\label{table:table1}
\end{table}


\subsection{Simulation from Mixture of Gaussian (Data Driven)}
\label{subsec: Data_Driven}

Here we show the performance of our procedure when we estimate the unknown parameters of distribution $b, \pi, \tau$ in model~\eqref{model} using algorithm~\ref{alg:algo1}. From now on in all the tables we denote ``$\boldsymbol{T}_N$-rule", ``Sun$\&$Cai", ``BH" and ``ABH", respectively,  the testing procedure based on $\{T_{i,N}, 1\le i \le K\}$ for $N=0,1,2$, the testing procedure developed in \cite{sun2007oracle}, the Benjamini-Hochberg procedure  \citep{benjamini1995controlling}, and the adaptive Benjamini-Hochberg, which incorporates the plug-in estimate of $\pi$ given by the method ``Est-EM" and ``Est-S\&C" (see, e.g., \citealt{benjamini2006adaptive}).

In Table~\ref{table:table2} by ``Est-EM" we mean estimating the parameters of the distribution using equations~\eqref{equation:est_fullem}. For AR(1) covariance, we take the subset $\{Z_2,Z_{6},\dots,Z_{3998}\}$ of $1000$ very weakly correlated z-scores and estimate the parameters using EM algorithm and in this case we get slightly inflated mFDR level. For Banded(1) covariance, we take the subset $\{Z_2,Z_4,Z_6,\dots,Z_{2000}\}$ of independent z-scores to estimate the parameters in this situation to get mFDR control for data driven rule. In this same table by ``Est-S\&C" we mean estimating the parameters of the distribution using equations~\eqref{equation:estimate_partialem}. 

 In Table \ref{table:table2} , we observe that 
the mFDR level is more conservative with  ``Est-S\&C" than with ``Est-EM". Due to this conservative nature of ``Est-S\&C", we make less discoveries in the case of ``Est-S\&C" compared to ``Est-EM". Nevertheless,  even with  ``Est-S\&C", using  the statistics $\boldsymbol{T}_N$ with $N=1,2$, results in much higher power than the power of $\boldsymbol{T}_0$ or the other competitors. 

\vspace{5mm}
\hspace{5mm}
\begin{table}
\begin{tabular}{|l|l|c|c|c|c|c|c|}
\hline
\multicolumn{8}{|c|}{Fixed Parameters in all Simulations $b=0$, $\tau = 2$} \\\hline
\multicolumn{1}{|c|}{Correlation} & \multicolumn{1}{|c|}{Error} & \multicolumn{1}{c|}{$\boldsymbol{T}_{0}$-rule} & \multicolumn{1}{c|}{$\boldsymbol{T}_{1}$-rule} & \multicolumn{1}{c|}{$\boldsymbol{T}_{2}$-rule} &   \multicolumn{1}{c|}{Sun$\&$Cai} & \multicolumn{1}{c|}{BH} & \multicolumn{1}{c|}{ABH}\\
\cline{2-8}
\hline

Banded(1),Est-EM& mFDR & 0.0494 & {0.0511} & 0.0495  & 0.0435 & 0.0327 & 0.0489\\
$\rho=0.5$,$\pi = 0.3$& FDR & 0.0486 & 0.0508 & 0.0494  & 0.0430 & 0.0324 & 0.0480\\
$K=2000$& TP & 124 & 174 & \bf{184} & 118 & 109 & 123\\
\hline

Banded(1),Est-EM & mFDR & 0.0497 & 0.0498 & {0.0509}  & 0.0475 & 0.0396 & 0.0497\\
$\rho=0.5$,$\pi = 0.2$ & FDR & 0.0480 & 0.0495 & 0.0506  & 0.0464 & 0.0383 & 0.0481\\
$K=2000$ & TP & 69 & 108 & \bf{118} & 68 & 65 & 70\\
\hline

AR(1),Est-EM& mFDR & {0.051} & {0.0526} & {0.0520}  & 0.0441 & 0.0356 & 0.0510\\
$\rho=0.5$,$\pi = 0.3$& FDR & 0.0503 & 0.0522 & 0.0517  & 0.0439 & 0.0355 & 0.0506\\
$K=4000$ & TP & 248 & 336 & \bf{344} & 237 & 219 & 249\\
\hline

AR(1),Est-EM& mFDR & {0.0522} & {0.0524} & {0.0525}  & 0.050 & 0.0416 & 0.0534\\
$\rho=0.5$,$\pi = 0.2$& FDR & 0.0511 & 0.0521 & 0.0521  & 0.0495 & 0.0411 & 0.0525\\
$K=4000$ & TP & 136 & 202 & \bf{206} & 134 & 127 & 137\\
\hline

AR(1),Est-S\&C & mFDR & 0.0342 & 0.0414 & 0.0417  & 0.0444 & 0.0333 & 0.0417\\
$\rho=0.5$,$\pi = 0.3$ & FDR & 0.0338 & 0.0413 & 0.0416  & 0.0440 & 0.0329 & 0.0413\\
$K=4000$ & TP & 218 & 313 & \bf{319} & 236 & 218 & 233\\
\hline

AR(1),Est-S\&C & mFDR & 0.0381 & 0.0423 & 0.0415  & 0.046 & 0.039 & 0.0447\\
$\rho=0.5$,$\pi = 0.2$,& FDR & 0.0375 & 0.042 & 0.0412 & 0.0455 & 0.0385 & 0.0441\\
$K=4000$ & TP & 124 & 191 & \bf{195} & 133 & 126 & 132\\
\hline
\end{tabular}
\caption{Power Gain with different covariance structure (Data Driven)\\
For different parameters of the model~\eqref{model} (rows),for each multiple testing procedure, we provide the FDR, mFDR, and power. The novel $T_N$-rules applied Algorithm \ref{alg:algo1} with  $B=50$.  
In bold, the most powerful procedure. 
Bootstrap standard error (s.e.) of all the mFDRs (for the combination ``AR(1),Est-EM") in $3$rd and $4$th rows are at most $0.0011$ and $0.0015$ respectively. Based on 200 data generations. } 
\label{table:table2}
\end{table}

\section{UK biobank height data analysis}
\label{sec:UKB}

Assume we have a standardized phenotype $(\boldsymbol{y}^{n\times 1})$ and standardized Genotype data $(X^{n \times p})$ available for $n$ individuals and $p$ different SNPs. We want to find the  SNPs that are associated with the phenotype, while  controlling  the mFDR at level 0.05,  using the $locFDR_N$ statistics for $N>0$. In our GWAS example we have a  sufficient number of individuals so  $n > p$. In order to apply our method we need z-scores and correlation among z-scores as input data. For finding the z-scores, we use the linear model,
$$\boldsymbol{y} = X\boldsymbol{\beta} + \boldsymbol{\epsilon}$$
and we know the least squares estimate of $\bb$ is $\hat{\bb} = (X^TX)^{-1}X^T\by$ with $\hat{\bb} \sim \mathcal{N}(\bb,\sigma^2(X^TX)^{-1})$.
Let $S = diag((X^TX)^{-1})$,  and let the estimated variance of $\epsilon_i$ be $\hat{\sigma^2} = \frac{1}{n-p}{(\by-X\hat{\bb})^T(\by-X\hat{\bb})}$. Our z-scores vector is $\bz = \frac{1}{\hat{\sigma}} S^{-1/2}\hat{\bb}$, and the correlation among the z-scores is 
$\Sigma = S^{-1/2}(X^TX)^{-1}S^{-1/2}$. So, our input data is $(\bz,\Sigma)$ and our working model is the two group model \eqref{model}.

We apply our method to some selected SNPs of chromosome 20, as described next.  We have the genotype data of around 500K individuals and approximately 20K SNPs for each of them from the UK biobank \citep{sudlow2015uk}. The individuals with missing height  (a small number) were removed from our analysis. In the raw data for each SNP we also have many missing individuals. Moreover, SNPs that are adjacent to each other tend to be highly correlated. 
  Therefore, we choose about 3.5K SNPs from chromosome 20 such that between any two of them there are at least 5 SNPs and none of these 3.5K SNPs has more than 10\% missing percentage. We impute the number of major allele count of each missing individuals for a particular SNP by the mean of all the available individuals' major allele count of that SNP. Next, we compute the $z$ scores vector and $\Sigma$, and use them in the  ``Est-S\&C" data driven procedure. The estimated mixture components are 
\begin{equation*}
    \hat{\pi} = 0.2,\quad \hat{b} = 0.0918,\quad \hat{\tau^2}=2.477
\end{equation*}
In Table~\ref{table:table3} we show simulation results with parameters estimated from the data $(\by,X)$ (essentially a parametric bootstrap experiment) after repeating the process $B$ times and aggregating the results. First row of the table shows the result where we assume the parameters are unknown and we estimate the parameters in each run (we take $B=200$ here). Second row of the table shows the result where we assume the parameters are known and we take $B=500$ here. Our procedures control the mFDR at the nominal 0.05 level.

The number of rejected SNPs for the real data at level $\alpha = 0.05$ by our method and the competitors with 3.5K SNPs are given  in Table~\ref{table:table4}. Though ``S$\&$C"  and ``$\boldsymbol{T}_0$-rule" use the same statistics for inference, due to different estimation methods of these statistics and different way of determining cutoffs, they reject different number of SNPs. The ``$\boldsymbol{T}_2$-rule" rejects $30$ more SNPs than the ``S$\&$C" method. 
Manhattan plots of all the methods are given in Figure~\ref{fig:fig5}. In all the subplots in Figure~\ref{fig:fig5}, we indicate by a red rectangle a specific region of the Genome where $T_{i,N}$ statistics with $N>0$ finds more associations than ``BH", ``ABH" and ``$\boldsymbol{T}_0$-rule". Note that in this specific region there are two associations ($1292$-th and $1405$-th SNP) 
which are found by all the methods. However between these two SNPs, there are two other close SNPs ($1350$-th and $1363$-th SNP) which are only identified by  $T_{i,N}$ statistics with $N>0$. Interestingly, we observe from the correlation matrix of the $z$-scores that these two newly identified SNPs are not highly correlated.
\vspace{5mm}
\hspace{0mm}
\begin{table}
\begin{tabular}{|l|l|c|c|c|c|c|c|}
\hline
\multicolumn{8}{|c|}{Parameters are $b = 0.0918$, $\tau^2 = 2.477$, $\pi = 0.2$ and $\Sigma = S^{-1/2}(X^TX)^{-1}S^{-1/2}$} \\

\hline
\multicolumn{1}{|c|}{Method} & \multicolumn{1}{|c|}{Error} & \multicolumn{1}{c|}{$\boldsymbol{T}_{0}$-rule} & \multicolumn{1}{c|}{$\boldsymbol{T}_{1}$-rule} & \multicolumn{1}{c|}{$\boldsymbol{T}_{2}$-rule} & 
\multicolumn{1}{c|}{Sun$\&$Cai} & \multicolumn{1}{c|}{BH} & \multicolumn{1}{c|}{ABH}\\ 
\cline{3-8}
\hline
Data & mFDR & 0.043 & 0.0449 &  0.0445 & 0.047 & 0.0408 & 0.045\\
Driven & FDR & 0.0419 & 0.044 &  0.0439 & 0.0468 & 0.0399 & 0.0438\\
 & TP & 52 & 59 & \bf{65} &  55 & 51 & 54\\
\hline
 & mFDR & 0.049 & {0.0501}  & 0.0488 & 0.0482 & 0.0408 & {0.0516}\\
Oracle & FDR & 0.0486 & 0.0499 & 0.0487 & 0.0467 & 0.0395 & 0.0497\\
 & TP & 57 & 65 & \bf{70} & 56 & 52 & 58\\
\hline
\end{tabular}
\caption{Power Gain by the simulation results using parameters estimated from real data \\
Boldsymbols indicates the best power among the competitors.}
\label{table:table3}
\end{table}

\vspace{5mm}
\hspace{5mm}
\begin{table}
\begin{tabular}{|l|c|c|c|c|c|c|c|}
\hline
\multicolumn{8}{|c|}{Number of Common Identified SNPs }\\
\hline
\multicolumn{1}{|c|}{Method} & \multicolumn{1}{c|}{Sun$\&$Cai} & \multicolumn{1}{c|}{BH} & \multicolumn{1}{c|}{ABH} &\multicolumn{1}{c|}{$\boldsymbol{T}_{0}$-rule} & \multicolumn{1}{c|}{$\boldsymbol{T}_{1}$-rule} & \multicolumn{1}{c|}{$\boldsymbol{T}_{2}$-rule} & \multicolumn{1}{c|}{$\boldsymbol{T}_{3}$-rule} \\
\cline{2-8}
\hline

S$\&$C & \bf{36} & 20 & 34 & 36 & 33 & 35 & 32 \\

\hline

BH & & \bf{20} & 20 & 20 & 19 & 20 & 20  \\

\hline

ABH & & & \bf{36} & 36 & 33 & 35 & 33 \\

\hline

$\boldsymbol{T}_{0}$-rule &&&& \bf{49} & 43 & 45 & 40\\

\hline

$\boldsymbol{T}_{1}$-rule &&&&& \bf{53} & 52 & 49\\

\hline

$\boldsymbol{T}_{2}$-rule &&&&&&\bf{66}&60\\

\hline

$\boldsymbol{T}_{3}$-rule &&&&&&&\bf{64}\\

\hline
\end{tabular}
\caption{Number of common identified SNPs by each of the method.
}
\label{table:table4}
\end{table}

\begin{figure} 
\begin{subfigure}{.5\textwidth}
  \centering
  \includegraphics[width=9cm, height=6cm]{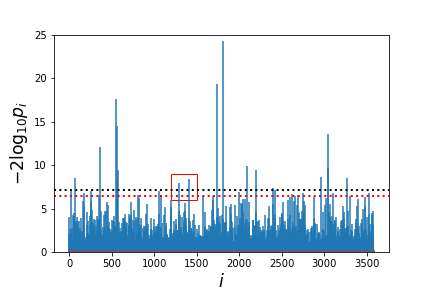}

  \label{fig:sub-first3}
\end{subfigure}
\begin{subfigure}{.5\textwidth}
  \centering
  \includegraphics[width=9cm, height=6cm]{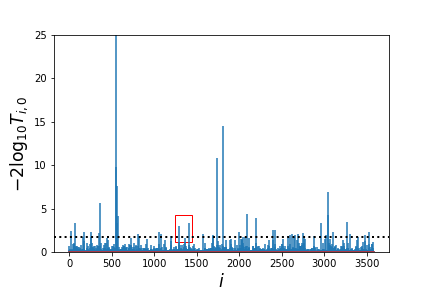}  
  
  \label{fig:sub-second3}
\end{subfigure}
\begin{subfigure}{.5\textwidth}
  \centering
  \includegraphics[width=9cm, height=6cm]{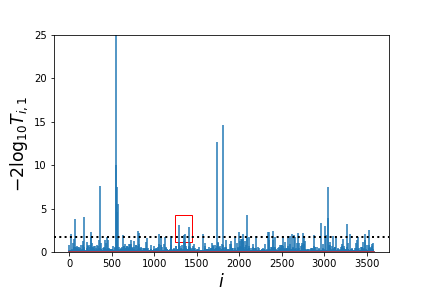}  
  
  \label{fig:sub-third3}
\end{subfigure}
\begin{subfigure}{.5\textwidth}
  \centering
  \includegraphics[width=9cm, height=6cm]{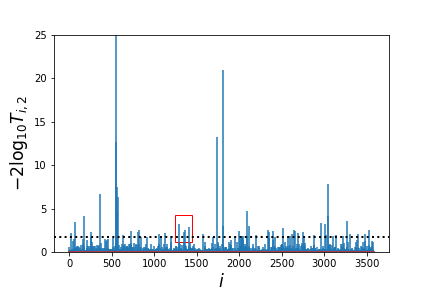}

  \label{fig:sub-fourth3}
\end{subfigure}

\caption{Manhattan plots of the analyzed SNPs\\
In the plot of ``$-2\log_{10}(p_i)$ vs $i$", black dotted line indicates the threshold by ``BH" procedure and red dotted line indicates the threshold by ``ABH" procedure\\
In the plot of ``$-2\log_{10}(T_{i,N})$ vs $i$", black dotted line indicates the threshold $-2\log_{10}(t_{\alpha,N})$
}
\label{fig:fig5}
\end{figure}


\section{Discussion}
\label{sec:disc}
In this paper we have studied practical multiple testing procedures in the two-group model under general dependence. We provide motivation and need for finding an optimal solution in a restricted decision space $\mathcal{C}_N$ and indeed show empirically that we do not lose much by concentrating on the smaller class $\mathcal{C}_N$. We provide practical approaches to estimate the parameters of a two group model under short range dependence structures which encompasses typical real applications. Using this approach we give a practical application in GWAS by finding more associations than  existing practical method using summary statistics. Our developed method works in GWAS without resorting to extensive modelling assumptions of complete bayesian framework as done in \cite{zhu2017bayesian}.

Though concentration on the class $\mathcal{C}_N$ of decision functions makes sense, it is natural to ask if there is better way of defining such a class which yields more practical and powerful solution. Our solution does not use the full correlation matrix -  it only uses the correlations within $N$ main diagonals of the correlation matrix. In other words, the solution ignores all the correlation outside the $N$-band. One possible approach can be to define $i$-th locFDR statistic by conditioning on top $N$ highest correlated $z$-scores and possibly we can choose this $N$ depending on $i$ determined by strength of correlation of the test statistic with the other test statistics. It will be interesting to explore how we can modify our method to accommodate the whole correlation matrix so we do not lose much ``information".

Proposition~\ref{prop:prop1} enables us to see how we can devise a method for mFDR control using the statistics $T_{i,N}$.  
If the entire joint distribution of the $z$-score vector is known, we can always define a policy of the form $\bD^t = (D_1^t(\bZ_{1,N}),\ldots,D_K^t(\bZ_{K,N}))$, $D_i^t(\bZ_{i,N}) = \mathbf{I}(T_{i,N}\le t)$ and determine $$t_{\alpha,N} = \sup\left\{t:Err(\bD^t)\le \alpha\right\}$$
with $Err\in \{FDR,pFDR,FDX\}$. This procedure $\bD^{t_{\alpha,N}}$ will control the chosen error by construction. However  we may not be able to prove similar optimality result like Theorem~\ref{theorem:theorem1} for the policy $\bD^{t_{\alpha,N}}$ with these error rates. To find optimal solutions, we need to solve the generic problem with $Err \in \{FDR,pFDR,FDX\}$,
\begin{equation}
\label{relaxed_rule_err}
            \max_{\bD \in \mathcal{C}_N}{TP(\bD)} \quad \text{s.t} \quad Err(\bD)\le\alpha.
\end{equation} 
The difficulty with finding such a solution is that other error measures such as FDX, FDR, pFDR do not decompose in the same way as mFDR as shown in equation~\eqref{eq:decompositions}. It is a challenge to find similar solution for these error measures.

 It is an interesting direction of future work to develop $T_{i,N}$ based method for $p>>n$ problems with properly chosen $z-$scores which takes into account the joint effect of covariates to some extent at least. 

\cite{sun2009large} developed a method assuming hypothesis states follow the hidden markov model (HMM) structure and the z-scores given the hypothesis states are independent. So their method does not use correlation among the $z$-scores. Interestingly,  proposition~\ref{prop:prop1} holds for any joint distribution of $(\bZ,\bh)$ on $\mathbb{R}^K \times \{0,1\}^K$. Although it is normally assumed that $h_i \overset{i.i.d.}{\sim} Ber(\pi)$, we can also assume that $\bh$ follows a HMM as done in \cite{sun2009large} and we can draw inference based on the statistics $T_{i,N}$ in these situations also. An interesting direction to extend our approach is thus to also allow an HMM structure of hypothesis states in addition to the correlation among $z$-scores.  We expect that it will be possible to derive a practical solution for this setting.

In particular, our estimation method is such that we can easily extend to this HMM based model as well, thereby producing a multiple testing procedure which we can directly apply to real data sets.

Our approach for parameter estimation in the two-group model relies on being able to extract a subset of test statistics which are roughly independent. We believe this is a realistic assumption in most real-life applications, in particular in GWAS. However it remains a challenge to design and test approaches that also apply in cases of long-range strong dependence. Note also that our estimation method do not take into account the known joint dependence among the statistics. Though this estimation approach is robust, taking joint dependence into account can lead to efficient estimate of the model parameters and in turn a boost in power. It may be interesting to explore how we can accommodate the correlations among the test statistics into estimation.

One major observation from our simulation results is that taking even part of the (short range) dependence into account  improves power. In fact, we observe large power gains even when there is very limited and not very strong dependence. Hence  our recommendation is  to use $locFDR_N$ statistics based rules, with $N>0$, whenever there is known dependence, or the dependence within neighborhoods can be well approximated.

\bigskip
\begin{center}
{\large\bf SUPPLEMENTARY MATERIAL}
\end{center}

\begin{description}

\item[Github Link of all relevant codes:] \url{https://github.com/rajeslab/locFDR_N} 

\item[Height data set:] We use Genotype ($X^{n\times p}$) and Height ($\by^{n\times 1}$) data of around 500K individuals provided by uk biobank. Unfortunately we can not make the data available publicly. For reproducibility, readers can use their own data in the notebook available in the above github link.

\end{description}

\bigskip
\begin{center}
{\large\bf Acknowledgments}
\end{center}

This study was supported by Israeli Science Foundation grant 2180/20. UK Biobank research has been
conducted using the UK Biobank Resource under Application Number 56885.

\bibliographystyle{apalike}
\bibliography{main.bib}
\appendix
\section{Proof of Theorem~\ref{theorem:theorem1}}\label{appendix:A}
Note that constraint $mFDR(\bD) \le \alpha$ is equivalent to $\mathbb{E}\left(V(\bD)-\alpha R(\bD)\right) \le 0$. The main essence of the proof relies on the following two decompositions of objective $TP(\bD)$ and the constraint $\mathbb{E}\left(V(\bD)-\alpha R(\bD)\right)$ for $\bD \in \mathcal{C}_N$:
\begin{align} \label{eq:decompositions}
    \mathbb{E}_{\bh,\bZ}\left(V(\bD)-\alpha R(\bD)\right) &=\sum_{i=1}^{K}{\mathbb{E}_{\bh,\bZ}\left[(1-h_i)D_i(\bZ_{i,N})-\alpha D_i(\bZ_{i,N})\right]} \nonumber \\ &=\sum_{i=1}^{K}{\mathbb{E}_{h_i,\bZ_{i,N}}\left[(1-h_i)D_i(\bZ_{i,N})-\alpha D_i(\bZ_{i,N})\right]} \nonumber \\&=\sum_{i=1}^{K}{\mathbb{E}_{\bZ_{i,N}}D_i(\bZ_{i,N})\mathbb{E}_{h_i|\bZ_{i,N}}\left[(1-h_i)-\alpha \right]}   \\&= \mathbb{E}_{\bZ}\left[\sum_{i=1}^{K}{D_i(\bZ_{i,N})(T_{i,N}-\alpha)}\right] \nonumber 
\end{align}
The following expression follows exactly the same way as above
\begin{align}
    TP(\bD) &= \mathbb{E}_{\bh,\bZ}\left(\sum_{i=1}^{K}{h_iD_i}\right) = \mathbb{E}_{\bZ}\left(\sum_{i=1}^{K}{D_i(\bZ_{i,N})(1-T_{i,N})}\right)
\end{align}

The lagrangian of the problem is
\begin{align}
    L(\bD,\mu) & = TP(\bD)+\mu\mathbb{E}(\alpha R(\bD)-V(\bD))\\
    &= \mathbb{E}_{\bZ}\left(\sum_{i=1}^{K}{D_i(\bZ_{i,N})(1-T_{i,N}+\alpha\mu-\mu T_{i,N})}\right)
\end{align}
Now the problem in \eqref{relaxed_rule} is equivalent to maximizing $L(\bD,\mu)$ over $\mathcal{C}_N$. Clearly the lagrangian is maximized for the rule of the form $\bD_N^t = (D_{1,N}^t,\dots,D_{K,N}^t)$ where
\begin{align}
    D_{i,N}^t(\bZ_{i,N}) &= \mathbb{I}\left(1-T_{i,N}+\alpha\mu-\mu T_{i,N} >0\right)\\
    &= \mathbb{I}\left(T_{i,N}<t\right)
\end{align}
with $t = \dfrac{1+\alpha\mu}{1+\mu}$. Now the constraint
\begin{equation}
\mathbb{E}\left(V(\bD_N^t)-\alpha R(\bD_N^t)\right) = \mathbb{E}_{\bZ}\left[\sum_{i=1}^{K}{\mathbb{I}\left(T_{i,N}<\dfrac{1+\alpha\mu}{1+\mu}\right)(T_{i,N}-\alpha)}\right] \le 0
\end{equation}
is satisfied if $\dfrac{1+\alpha\mu}{1+\mu} \le \alpha$. Note that $TP(\bD_N^t)$ is increasing in $t = \dfrac{1+\alpha\mu}{1+\mu}$ and hence the optimal solution is the decision $\bD_N^{t_{\alpha,N}}=\bD_N^*$ where
\begin{align*}
    t_{\alpha,N} &= \sup\left\{t: \mathbb{E}\left(V(\bD_N^t)-\alpha R(\bD_N^t)\right) \le 0\right\} \\
    &=\sup \left\{t:\dfrac{\mathbb{E}\sum_{i = 1}^{K}{\mathbb{I}(T_{i,N}\le t)T_{i,N}}}{\mathbb{E}\sum_{i = 1}^{K}{\mathbb{I}(T_{i,N}\le t)}}\le \alpha\right\}
\end{align*}
\section{Derivation of the equations in \eqref{equation:est_fullem}}\label{appendix:B}
Let us assume that we have the model \eqref{eq-mixturemodel}. Our task is to estimate the unknown parameters of the model using EM algorithm. We consider $\bz = (z_1,\dots,z_{\Tilde{K}})$ as our observed data and $\bh = (h_1,\dots, h_{\Tilde{K}})$ as unobserved data. Log likelihood of the complete data $(\bz,\bh)$ is 

\begin{align}
    \log L(\bz,\bh,\bt) & = \log \left[p(\bz|\bh,\bt)p(\bh|\bt)\right]\\
    &= \log \left[\prod_{i}{\dfrac{1}{\sqrt{2\pi(1+\tau^2h_i)}}\exp{-\left[\dfrac{1}{2(1+\tau^2h_i)}(z_i - bh_i)^2\right]}\pi^{h_i}(1-\pi)^{(1-h_i)}}\right] \\
    & = c - \dfrac{1}{2}\sum_{i}{\log (1+\tau^2h_i)}-\dfrac{1}{2}\sum_{i}{\dfrac{1}{(1+\tau^2h_i)}(z_i - bh_i)^2}\\&+\log \pi\sum_{i}{h_i}+\log(1-\pi)\sum_{i}{(1-h_i)}
\end{align}

\textbf{E Step:} Let us assume that the estimate of the parameters in the $t$-th step is $\bt^{(t)}=(\pi^{(t)},b^{(t)},\tau^{2(t)})$. Then we need to find the estimate at the $t+1$-th step by maximizing the conditional expectation $\mathbb{E}\left[\log L(\bz,\bh,\bt)\right\vert\bz,\bt^{(t)}]$.

\textbf{M-Step:}
Clearly from the above expression of $\log L(\bz,\bh,\bt)$
\begin{gather*}
    \dfrac{\delta}{\delta\pi}\mathbb{E}\left[\log L(\bz,\bh,\bt)\right\vert\bz,\bt^{(t)}] = \dfrac{1}{\pi}\sum_{i}{\mathbb{E}(h_i\vert z_i,\bt^{(t)})}-\dfrac{1}{1-\pi}\sum_{i}{(1-\mathbb{E}(h_i\vert z_i,\bt^{(t)}))}=0\\
    \implies \pi^{(t+1)} = \dfrac{1}{\Tilde{K}}\sum_{i}{\mathbb{E}(h_i\vert z_i,\bt^{(t)})} = \dfrac{1}{\Tilde{K}}\sum_{i}{\mathbb{P}(h_i = 1 \vert z_i,\boldsymbol{\theta}^{(t)})} 
\end{gather*}
Similarly setting $\dfrac{\delta}{\delta b}\mathbb{E}\left[\log L(\bz,\bh,\bt)\right\vert\bz,\bt^{(t)}] = 0$, we get
\begin{gather*}
    \dfrac{\delta}{\delta b}\left[\sum_{i}{(z_i-b)^2\mathbb{P}(h_i = 1 \vert z_i,\boldsymbol{\theta}^{(t)})}\right]=0\\
    \implies b^{(t+1)} = \dfrac{\sum_{i}\mathbb{P}(h_i = 1 \vert z_i,\boldsymbol{\theta}^{(t)})z_i}{\sum_{i}\mathbb{P}(h_i = 1 \vert z_i,\boldsymbol{\theta}^{(t)})}
\end{gather*}
Now for determining $\tau^{2(t+1)}$ we set $\dfrac{\delta}{\delta \tau^2}\mathbb{E}\left[\log L(\bz,\bh,\pi,b^{(t+1)},\tau)\right\vert\bz,\bt^{(t)}] = 0$ to get
\begin{gather*}
    \sum_{i}{\dfrac{1}{(1+\tau^2)^2}(z_i - b^{(t+1)})^2 \mathbb{P}(h_i = 1 \vert z_i,\boldsymbol{\theta}^{(t)})} = \sum_{i}{\dfrac{1}{(1+\tau^2)} \mathbb{P}(h_i = 1 \vert z_i,\boldsymbol{\theta}^{(t)})}\\
    \implies {\tau^2}^{(t+1)} =\max\left\{0, \dfrac{\sum_{i}\mathbb{P}(h_i = 1 \vert z_i,\boldsymbol{\theta}^{(t)})(z_i-b^{(t+1)})^2}{\sum_{i}\mathbb{P}(h_i = 1 \vert z_i,\boldsymbol{\theta}^{(t)})}-1\right\}
\end{gather*}
Hence the updates of equations~\eqref{equation:est_fullem} follows. Updates of the equations~\eqref{equation:estimate_partialem} follows similarly.
\end{document}